\documentclass{llncs}

\usepackage{times}
\usepackage{algorithm2e}
\usepackage{graphicx}
\usepackage{gastex}
\usepackage{wrapfig}
\usepackage{paralist}
\usepackage{multirow}
\usepackage{paralist}
\usepackage[T1]{fontenc}
\usepackage[scaled=0.9]{DejaVuSansMono}
\usepackage{listings}
\let\llncssubparagraph\subparagraph
\let\subparagraph\paragraph
\let\subparagraph\llncssubparagraph

\usepackage{amssymb}
\usepackage{amsmath}
\usepackage{url}
\usepackage{color}

\newcommand{\Inv}{\mathit{Inv}}
\newcommand{\I}{\mathit{I}}

\SetKwFunction{CG}{ConvexGeneralizer}
\SetKwFunction{FG}{FarkasGeneralizer}

\newcommand{\limp}{\Rightarrow}

\newcommand{\sat}{\textsc{sat}}

\newcommand{\Queue}{\mathcal{Q}}

\newcommand{\Bad}{\mathit{Bad}}
\newcommand{\Init}{\mathit{Init}}

\newcommand{\Spacer}{\textsc{Spacer}\xspace}
\newcommand{\spacer}{\Spacer}

\newcommand{\Quic}{\textsc{Quic3}\xspace}
\newcommand{\quic}{\Quic}

\newcommand{\papercomment}[1]{}

\newcommand{\nat}{\mathbb{N}}

\newcommand{\cF}{\mathcal{F}}

\newcommand{\pMBP}{\textsc{pMbp}}

\newcommand{\Itp}{\textsc{Itp}}
\newcommand{\pItp}{\textsc{pItp}}

\newcommand{\Tr}{\mathit{Tr}}

\newcommand{\st}{\mathbin{\cdot}}
\newcommand{\cT}{\mathcal{T}}
\newcommand{\sel}{\mathsf{sel}\xspace}
\newcommand{\store}{\mathit{store}}

\newcommand{\Int}{\textsf{int}}
\newcommand{\Array}{\textsf{array}}

\newcommand{\var}[1]{v_{#1}}
\newcommand{\dom}{\mathit{dom}}
\newcommand{\range}{\mathit{range}}
\newcommand{\terms}{\mathit{Terms}}
\newcommand{\fvars}{\mathit{FVars}}

\newcommand{\const}{\mathit{Const}}

\newcommand{\abs}{\mathit{abs}}
\newcommand{\SK}{\mathit{SK}}
\newcommand{\sk}{\mathit{sk}}
\newcommand{\qi}{\mathit{qi}}

\newcommand{\BMC}{\mathit{BMC}}
\newcommand{\BMCk}[1]{\BMC_{#1}}

\newcommand{\qgen}{\textbf{QGen}\xspace}

\renewcommand{\paragraph}[1]{\vspace{0.2cm}\noindent {\normalfont\normalsize\itshape #1}}

\usepackage{hyperref}
\usepackage[capitalise]{cleveref}

\usepackage{listings}
\lstdefinestyle{CStyle}{
    language=C,
    showtabs=false,
    tabsize=2,
    basicstyle=\footnotesize
}
\usepackage{verbatim}

\Crefname{algorithm}{Alg.}{Alg.}

\author{Arie Gurfinkel\inst{1} \and Sharon Shoham\inst{2} \and Yakir Vizel\inst{3}}

\title{Quantifiers on Demand}

\institute{University of Waterloo \and Tel Aviv University \and The Technion}

\begin{document}
\pagestyle{empty}
\maketitle

\begin{abstract}
  Automated program verification is a difficult problem. It is
  undecidable even for transition systems over Linear Integer
  Arithmetic (LIA). Extending the transition system with theory of Arrays,
  further complicates the problem by requiring inference and reasoning
  with universally quantified formulas. In this paper, we present a
  new algorithm, \textsc{Quic3}, that extends IC3 to infer universally
  quantified invariants over the combined theory of LIA and Arrays. Unlike other approaches that use either IC3
  or an SMT solver as a black box, \textsc{Quic3} carefully manages
  quantified generalization (to construct quantified invariants) and
  quantifier instantiation (to detect convergence in the presence of
  quantifiers). While \textsc{Quic3} is not guaranteed to converge, it
  is guaranteed to make progress by exploring longer and longer
  executions. We have implemented \textsc{Quic3} within the Constrained
  Horn Clause solver engine of Z3 and experimented with it by applying
  \textsc{Quic3} to verifying a variety of public benchmarks of array
  manipulating C programs.
\end{abstract}

\setlength{\textfloatsep}{5pt}

\section{Introduction}
\label{sec:intro}

Algorithmic logic-based verification (ALV) is one of the most
prominent approaches for automated verification of software. ALV
approaches use SAT and SMT solvers to reason about bounded program
executions; and generalization techniques, such as interpolation, to
lift the reasoning to unbounded executions. In recent years,
IC3~\cite{DBLP:conf/vmcai/Bradley11} (originally proposed for hardware
model checking) and its extensions to Constrained Horn Clauses (CHC)
over SMT
theories~\cite{DBLP:conf/sat/HoderB12,DBLP:conf/fmcad/KomuravelliBGM15}
has emerged as the most dominant ALV technique. The efficiency of the
IC3 framework is demonstrated by success of such verification tools as
\textsc{SeaHorn}~\cite{DBLP:conf/cav/GurfinkelKKN15}.

The IC3 framework has been successfully extended to deal with
arithmetic~\cite{DBLP:conf/sat/HoderB12}, arithmetic and
arrays~\cite{DBLP:conf/fmcad/KomuravelliBGM15}, and universal
quantifiers~\cite{DBLP:conf/cav/KarbyshevBIRS15}. However, no
extension supports the \emph{combination} of all three. Extending IC3
to Linear Integer Arithmetic (LIA), Arrays, and
Quantifiers is the subject of this paper. Namely, we present a technique
to discover universally quantified solutions to CHC over the theories of
LIA and Arrays. These solutions correspond to universally quantified
inductive invariants of array manipulating programs.

For convenience of presentation, we present our approach over a
transition system modelled using the theories of Linear Integer
Arithmetic (LIA) and Arrays, and not the more general, but less
intuitive, setting of CHCs. Inductive invariants of such transition
systems are typically quantified, which introduces two major
challenges:
\begin{inparaenum}[(i)]
\item quantifiers tremendously increase the search space for a candidate
  inductive invariant, and
\item they require deciding satisfiability of quantified formulas --
  itself an undecidable problem.
\end{inparaenum}

Existing ALV techniques for inferring universally quantified arithmetic
invariants either restrict the shape of the quantifiers and
reduce to quantifier free
inference~\cite{DBLP:conf/sas/BjornerMR13,DBLP:conf/sas/MonniauxG16,fse16},
or guess quantified invariants from bounded
executions~\cite{DBLP:conf/cav/AlbertiBGRS12}.

In this paper, we introduce \textsc{Quic3} -- an extension of
\textsc{IC3}~\cite{DBLP:conf/vmcai/Bradley11,DBLP:conf/sat/HoderB12,DBLP:conf/cav/KomuravelliGC14}
to universally quantified invariants. Rather than fixing the shape of
the invariant, or discovering quantifiers as a post-processing phase,
\textsc{Quic3} computes the necessary quantifiers \emph{on demand} by
taking quantifiers into account during the search for invariants.  The
key ideas are to allow existential quantifiers in proof obligations
(or, counterexamples to induction) so that they are \emph{blocked} by
universally quantified lemmas, and to extend lemma generalization to
add quantifiers.

Generating quantifiers on demand gives more control over the
inductiveness checks. These checks (i.e., pushing in \textsc{IC3})
require deciding satisfiability of universally quantified formulas
over the combined theory of Arrays and LIA. This is undecidable, and
is typically addressed in SMT solvers by \emph{quantifier
  instantiation} in which a universally quantified formula
$\forall x \cdot \varphi(x)$ is approximated by a finite set of ground
instances of $\varphi$. SMT solvers, such as
Z3~\cite{DBLP:conf/tacas/MouraB08}, employ sophisticated
heuristics~(e.g.,~\cite{DBLP:conf/cav/GeM09}) to find a sufficient set
of instantiations. However, the heuristics are only complete in
limited situations (recall, the problem is undecidable in general),
and it is typical for the solver to return \emph{unknown}, or, even
worse, diverge in an infinite set of instantiations.

Instead of using an SMT solver as a black-box, \textsc{Quic3}
generates and maintains a set of instantiations on demand. This
ensures that \textsc{Quic3} always makes progress and is never stuck
in a single inductiveness check. The generation of instances is driven
by the \emph{blocking} phase of \textsc{IC3} and is supplemented by
traditional pattern-based triggers.
Generating both universally quantified lemmas and their instantiations
on demand, driven by the property, offers additional flexibility
compared to the eager quantifier instantiation approach
of~\cite{DBLP:conf/sas/BjornerMR13,DBLP:conf/sas/MonniauxG16,fse16}.

Combining the search for all of the ingredients (quantified and
quantifier-free formulas, and instantiations) in a single procedure
improves the control over the verification process. For example, even
though there is no guarantee of convergence (the problem is, after
all, undecidable), we guarantee that \textsc{Quic3} makes progress,
exploring more of the program, and discovering a counter-example (even
the shortest one) if it
exists.

While our intended target is program verification, we have implemented
\textsc{Quic3} in a more general setting of Constrained Horn Clauses
(CHC). We build on the Generalized PDR
engines~\cite{DBLP:conf/sat/HoderB12,DBLP:conf/cav/KomuravelliGC14} in
Z3. The input is a set of CHC in SMT-LIB format, and the output is a
universally quantified inductive invariant, or a counter-example. To
evaluate \textsc{Quic3}, we have used array manipulating C programs
from SV-COMP. 
We show that our implementation is competitive and can automatically
discover non-trivial quantified invariants.

In summary, the paper makes the following contributions: (a) extends IC3 framework to support quantifiers;
(b) develops quantifier generalization techniques; (c) develops
techniques for discovering quantifier instantiations during
verification; and (d) reports on our
implementation for software verification.

\vspace{-0.1in}
\section{Preliminaries}
\label{sec:preliminaries}
\vspace{-0.2in}

\paragraph{Logic.}
 We consider First Order Logic modulo the combined theory of Linear Integer Arithmetic (LIA) and Arrays.
We denote the theory by $\cT$ and the logic by $FOL(\cT)$.
We assume that the reader is familiar with the basic notions of $FOL(\cT)$
and provide only a brief description to set the notation.
Formulas in $FOL(\cT)$ are defined over a signature
$\Sigma$ which includes sorts $\Int$ and $\Array$, where sort $\Int$ is also used as the sort of the array indices and data.
We assume that the signature $\Sigma$ includes equality (=),
interpreted functions, predicates, and constants of arithmetic (i.e.,
the functions $+$, $-$, $*$, the predicates $<$, $\leq$, and the
constants $1$, $2$, etc.) and of arrays (i.e., the functions $\sel$
and $\store$).

In addition, $\Sigma$ may be extended with uninterpreted constants. In
particular, we assume that $\Sigma$ includes special \emph{Skolem}
uninterpreted constants $\SK = \{sk_i\}$ of sort {\Int} for $i$ in
natural numbers.

We denote by $\Sigma_{\cT}$ the interpreted part of $\Sigma$, and by $X \subseteq \Sigma$ the set of uninterpreted constants (e.g., $a$ or $sk_i$, but not $1$).
In the sequel we write $\varphi(X)$, and say that $\varphi$ is defined over $X$, to denote that $\varphi$ is defined over signature $\Sigma = \Sigma_{\cT} \cup X$.
We write $\const(\varphi) \subseteq X$ for the set of all uninterpreted constants that appear in $\varphi$.
In the rest of the paper, whenever we refer to constants, we only refer to the uninterpreted ones.

We write $T$ for the set of terms of $FOL(\cT)$, and $V$ for the set
of (sorted) variables. We assume that {\Int} variables in $V$ are of the form
$\var{i}$, where $i$ is a natural number. Thus, we can refer to all
such variables by their numeric name.  For a formula $\varphi$, we write
$\terms(\varphi) \subseteq T$ and $\fvars(\varphi) \subseteq V$ for
the terms and free variables of $\varphi$, respectively.

A substitution $\sigma: V \to T$ is a partial mapping from $V$ to
terms in $T$ that pertains to the sort constraints. We write $\dom(\sigma)$ to denote the domain of
$\sigma$, and $\range(\sigma)$ to denote its range.
For a formula $\varphi$, we write $\varphi \sigma$ for the
result of applying substitution $\sigma$ to $\varphi$. Abusing
notation, we write $\emptyset$ for an empty substitution, i.e.,
a substitution $\sigma$ such that $\dom(\sigma) = \emptyset$. Given two
substitutions $\sigma_1$ and $\sigma_2$, we write $(\sigma_1 \mid
\sigma_2)$ for a composition of substitutions defined such that:
$(\sigma_1 \mid \sigma_2)(x) = \sigma_1(x)$ if $x \in\dom(\sigma_1)$,
and $\sigma_2(x)$, otherwise.
We define a special \emph{Skolem substitution} $\sk:V \to T$ such that
$\sk(\var{i}) = \sk_i$ for $\sk_i \in \SK$. Given a formula $L$, we
write $L_{\sk}$ for $L\sk$, and given a substitution~$\sigma$. 

We write $\abs (U, \varphi) = (\psi, \sigma)$ for an abstraction
function that given a set of uninterpreted constants $U$ and a formula
$\varphi$ returns an abstraction $\psi$ of $\varphi$ in which the
constants are replaced by free variables, as well as a substitution
$\sigma$ that records the mapping of variables back to the constants
that they abstract. Formally, we require that
$\abs (U, \varphi) = (\psi, \sigma)$ satisfies the following:
$\psi \sigma = \varphi$,
$\dom(\sigma) = \fvars(\psi) \setminus \fvars(\varphi)$, and
$U \cap \terms(\psi) = \emptyset$.
The requirements ensure that $\abs$ abstracts all uninterpreted constants in
$U$, and $\sigma$ maps the newly introduced variables back to the
constants. 
Furthermore, we require that for every skolem constant $\sk_i$ in $U$, $\abs(U, \varphi)$ abstracts $\sk_i$ in $\varphi$ to $\var{i}$ in $\psi$, and accordingly, $\sigma(\var{i})=\sk_i$.
This ensures that applying skolemization, followed by abstraction of $\SK$, reintroduces the same variables and does not result in variable renaming.
That is, $\abs(\SK, \varphi_{\sk}) = (\varphi, \_)$.

We write $\forall \varphi$ for a formula obtained from $\varphi$ by
universally quantifying all free variables of $\varphi$, and
$\exists \varphi$ for a formula obtained by existential
quantification, respectively. For convenience, given a set of
constants $U$ and a ground formula $\varphi$ (i.e., a formula where
all terms are ground), we write $\exists U \st \varphi$ for
$\exists \psi$, where $(\psi, \sigma) = \abs(U, \varphi)$.
We write $\varphi \limp \psi$ do denote the validity of $\varphi \rightarrow \psi$.

\paragraph{Model Based Projection.} Given a ground formula $\varphi$,
a model $M$ of $\varphi$, and a set of uninterpreted constants
$U \subseteq \const(\varphi)$,
 (partial, or incomplete) Model Based Projection,
MBP, is a function $\pMBP (U, \varphi, M) = (\psi, W)$ such that
\begin{inparaenum}
\item $\psi$ is a ground monomial (i.e., conjunction of ground
  literals),
\item $W \subseteq U$ and $\const(\psi) \subseteq \const(\varphi) \setminus (U \setminus W)$,
\item $\psi \limp (\exists U \setminus W \st \varphi)$,
\item $M \models \psi$,
\item $\pMBP$ is finite ranging in its third argument: for a fixed $U$
  and $\varphi$, the set
  $\{\pMBP (U, \varphi, M) \mid M \models \varphi \}$ is finite.
\end{inparaenum}
Intuitively, the monomial $\psi$ underapproximates (implies) the
result of eliminating the existential quantifiers pertaining to
$U \setminus W$ from $\varphi$ (where quantifier elimination itself
may not even be defined). It, therefore, represents one of the ways of
satisfying the result of quantifier elimination. The
underapproximation $\psi$ is chosen such that it is consistent with
the provided model $M$. In this paper,  MBP is used as a
way to underapproximate the pre-image of a set of states represented
implicitly by some formula.

An MBP is called \emph{complete} if $W$ is always empty. A complete
MBP for Linear Arithmetic has been presented
in~\cite{DBLP:conf/cav/KomuravelliGC14} and a partial MBP for the
theory of arrays has been presented
in~\cite{DBLP:conf/fmcad/KomuravelliBGM15}.  Importantly, in the
partial MBP of~\cite{DBLP:conf/fmcad/KomuravelliBGM15}, the remaining
set of constants, $W$, never contains any constant of sort \Array.  We
refer the readers
to~\cite{DBLP:conf/cav/KomuravelliGC14,DBLP:conf/fmcad/KomuravelliBGM15}
and to~\cite{DBLP:conf/lpar/BjornerJ15} for details. A complete MBP
under-approximates quantifier elimination relative to a given
model. Such an MBP can only exist if the underlying theory admits
quantifier elimination. Since the theory of arrays does not admit
quantifier elimination it only admits a partial MBP.

In the paper, we further require an MBP to eliminate all the constants
of sort {\Array} from $U$, such as the MBP
of~\cite{DBLP:conf/fmcad/KomuravelliBGM15}.

\paragraph{Interpolation.} Given a ground formula $A$, and a ground
monomial $B$ such that $A \limp \neg B$, (partial) interpolation,
ITP, is a function $\pItp(A, B) = (\varphi, U)$, s.t.
\begin{inparaenum}
\item $\varphi$ is a ground clause (i.e., a disjunction of ground literals),
\item $U \subseteq \const(B) \setminus \const(A)$
  and $\const(\varphi) \subseteq (\const(A) \cap \const(B)) \cup U$,
\item $A \limp \forall U \st \varphi$, and
\item $\varphi \limp \neg B$.
\end{inparaenum}
The set of constants $U$ denotes the constants of $\varphi$ that exceed the set of shared constants of $A$ and $B$.
An interpolation procedure is complete if for any pair $A$, $B$, the
returned set $U$ is always empty. The formula $\varphi$ produced by a
complete interpolation procedure is called an \emph{interpolant} of
$A$ and $B$. Note that our definitions admit a trivial partial
interpolation procedure defined as
$\pItp_{\mathit{triv}} (A, B) = (\neg B, \const(B) \setminus
  \const (A))$.

\paragraph{Safety problem.}  We represent transition systems via
formulas in $FOL(\cT)$. The states of the system correspond to
structures over a signature $\Sigma = \Sigma_{\cT} \cup X$, where $X$
denotes the set of (uninterpreted) constants.  The constants in $X$
are used to represent program variables.  A \emph{transition system}
is a pair $\langle \Init(X), \Tr(X,X') \rangle$, where $\Init$ and
$\Tr$ are quantifier-free ground formulas in $FOL(\cT)$.  $\Init$
represents the initial states of the system and $\Tr$ represents the
transition relation. We write $\Tr(X,X')$ to denote that $\Tr$ is
defined over the signature $\Sigma_{\cT} \cup X \cup X'$, where $X$ is
used to represent the pre-state of a transition, and
$X' = \{a' \mid a \in X\}$ is used to represent the post-state.  A
\emph{safety problem} is a triple
$\langle \Init(X), \Tr(X,X'), \Bad(X) \rangle$, where
$\langle \Init, \Tr \rangle$ is a transition system and $\Bad$ is a
quantifier-free ground formula in $FOL(\cT)$ representing a set of
bad states.

The safety problem $\langle \Init(X), \Tr(X,X'), \Bad(X)\rangle$ has a \emph{counterexample of length $k$} if the following formula is satisfiable:
\[
\BMCk{k}(\Init,\Tr,\Bad) = \Init(X_0) \wedge \bigwedge_{i=0}^{k-1} \Tr(X_i,X_{i+1}) \wedge \Bad(X_k),
\]
where $X_i = \{a_i \mid a \in X\}$ is a copy of the constants used to represent the state of the system after the execution of $i$ steps.
The transition system is \emph{safe} if the safety problem has no counterexample, of any length.

\paragraph{Interpolation sequence and inductive invariants.} An \emph{interpolation sequence of length $k$} for a safety problem $\langle \Init(X), \Tr(X,X'), \Bad(X) \rangle$ is a sequence of formulas $\I_1(X),\ldots,\I_k(X)$ such that
\begin{inparaenum}[(i)]
\item $\Init(X) \limp \I_1(X)$,
\item $\I_j(X) \wedge \Tr(X,X') \limp \I_{j+1}(X')$ for every $1 \leq j \leq k-1$, and
\item $\I_k(X)\limp \neg \Bad(X)$.
\end{inparaenum}
If an interpolation sequence of length $k$ exists, then the transition system has no counterexample of length $k$.
An \emph{inductive invariant} is a formula $\Inv(X)$ such that
\begin{inparaenum}[(i)]
\item $\Init(X) \limp \Inv(X)$,
\item $\Inv(X) \wedge \Tr(X,X') \limp \Inv(X')$, and
\item $\Inv(X)\limp \neg \Bad(X)$.
\end{inparaenum}
If such an inductive invariant exists, then the transition system is safe.

\section{Quantified IC3}
\label{sec:linear-quic3}

\begin{algorithm}[t]\small
  \KwIn{A safety problem
    $\langle \Init(X), \Tr(X,X'), \Bad(X) \rangle$.}

  \KwSty{Assumptions}: $\Init$, $\Tr$ and $\Bad$ are quantifier free.

  \KwData{A POB queue $\Queue$, where a POB $c \in \Queue$ is a triple
    $\langle m, \sigma, i \rangle$, $m$ is a conjunction of literals
    over
    $X$ 
    and free variables, $\sigma$ is a substitution s.t. $m\sigma$ is
    ground, and $i \in \nat$. A level $N$. A quantified trace
    $\cT = Q_0, Q_1, \ldots$, where for every pair
    $(\ell, \sigma) \in Q_i$, $\ell$ is a quantifier-free formula over
    $X$ and free variables and $\sigma$ a substitution
    s.t. $\ell\sigma$ is ground.}

  \KwSty{Notation}: $\cF(A) = (A(X) \land \Tr(X,X')) \lor
  \Init(X')$; $\qi(Q) = \{\ell\sigma \mid (\ell, \sigma) \in
  Q\}$; $\forall Q = \{\forall \ell \mid (\ell, \sigma) \in Q \}$.
  \\[0.1in]

  \KwOut{\emph{Safe} or \emph{Cex} }

  \KwSty{Initially:}
  $\Queue = \emptyset$, $N=0$, 
  $Q_0=\{(\Init,\emptyset)\}$,
  $\forall i > 0 \st Q_i = \emptyset$.\\
  \Repeat{$\infty$} {
    \begin{description}
      \setlength\itemsep{0.05in}
    \item[Safe] If there is an $i < N$ s.t.
    $\forall Q_{i} \subseteq \forall Q_{i+1}$
      \Return{\textit{Safe}}.
    \item[Cex] If there is an $m, \sigma$ s.t. $\langle m, \sigma, 0 \rangle \in
      \Queue$  \Return{\textit{Cex}}.
    \item[Unfold] If 
     $\qi(Q_N) \limp \neg \Bad$, then set
      $N \gets N + 1$.
    \item[Candidate]
      If for some $m$, 
      $m \limp \qi(Q_N) \land \Bad$,
      then add $\langle m, \emptyset, N\rangle$ to $\Queue$.
    \item[Predecessor] If $\langle m, \xi, i+1\rangle \in \Queue$ and
      there is a model $M$ s.t.
      $M \models \qi(Q_i) \land \Tr \land (m'_{\sk})$, add
      $\langle \psi, \sigma, i \rangle$ to $\Queue$, where
      $(\psi, \sigma) = \abs(U, \varphi)$ and
      $(\varphi, U) = \pMBP (X'\cup \SK, \Tr \land m'_{\sk}, M)$.
    \item[NewLemma] For $0 \leq i < N$, given a POB
      $\langle m, \sigma, i+1\rangle\in\Queue$ s.t.
      $\cF(\qi(Q_i)) \land m'_{\sk}$ is unsatisfiable, and
      $L' = \Itp (\cF(\qi(Q_i)), m'_{\sk})$,
      add $(\ell, \sigma)$ to $Q_j$ for
      $j \leq i + 1$, where $(\ell, \_ ) = \abs(\SK, L)$.
  \item[Push]
  For $0 \leq i < N$ and
    $((\varphi \lor\psi), \sigma) \ \in Q_i$, if
    $(\varphi, \sigma)\not\in Q_{i+1}$, $\Init \limp \forall\varphi$ and
    $(\forall\varphi) \land \forall Q_i \land \qi(Q_i) \land \Tr \;\limp\; \forall\varphi'$, then
    add $(\varphi, \sigma)$ to $Q_{j}$, for all $j \leq i + 1$.

    \end{description}
  }
  \caption{The rules of \textsc{Quic3} procedure.}
  \label{alg:lquic} 
  \vspace{-0.1in}
\end{algorithm}

In this section, we present \textsc{Quic3} -- a procedure for
determining a safety of a transition system by inferring quantified
inductive invariants.  Given a safety problem, \textsc{Quic3} attempts
to discover an inductive invariant $\Inv(X)$ as a
universally-quantified formula of $FOL(\cT)$ (where quantification is
restricted to variables of sort $\Int$) or produce a counterexample.

We first present \textsc{Quic3} as a set of rules, following the
presentation style
of~\cite{DBLP:conf/sat/HoderB12,DBLP:conf/cav/KomuravelliGC14,DBLP:conf/vmcai/BjornerG15,DBLP:conf/fmcad/KomuravelliBGM15,DBLP:conf/fmcad/GurfinkelI15}. We
focus on the data structures, the key differences between
\textsc{Quic3} and \textsc{IC3}, and soundness of the rules. An
imperative procedure based on these rules is presented in
\cref{sec:progress-cex}. We assume that the reader is familiar with
the basics of \textsc{IC3}.
Throughout the section, we fix a safety problem
$P = \langle \Init(X), \Tr(X,X'), \Bad(X) \rangle$, and assume that
$\Init$, $\Tr$ and $\Bad$ are quantifier free ground formulas. For
convenience of presentation, we use the notation $\cF(A)$ to denote
the formula $(A(X) \land \Tr(X,X')) \lor \Init(X')$ that corresponds
to the forward image of $A$ over the $\Tr$ extended by the initial
states.

The rules of \textsc{Quic3} are shown in \cref{alg:lquic}. Similar to
\textsc{IC3}, \textsc{Quic3} maintains a queue $\Queue$ of proof
obligations (POBs), and a monotone inductive trace $\cT$ of frames
containing lemmas at different levels. However, both the proof
obligations and the lemmas maintained by \textsc{Quic3} are
quantified.

\paragraph{Quantified Proof Obligations.} Each POB in $\Queue$ is a
triple $\langle m, \sigma, i \rangle$, where $m$ is a
monomial 
over $X$ such that $\fvars(m)$ are of sort $\Int$, $\sigma$ is a
substitution such that
$\fvars(m) \subseteq \dom(\sigma)$ and
$\range(\sigma) \subseteq X' \cup \SK$, and $i$ is a natural number
representing the frame index at which the POB should be either blocked
or extended.  The POB $\langle m, \sigma, i \rangle$ expresses an
obligation to show that no state satisfying $\exists m$ is reachable
in $i$ steps of $\Tr$. The substitution $\sigma$ records the specific
instance of the free variables in frame $i+1$ that were abstracted
during construction of $m$.  Whenever the POB is blocked, a
universally quantified lemma $\forall\ell$ is generated in frame $i$
(as a generalization of $\forall \neg m$), and, $\sigma$ is used to
discover the specific instance of $\forall\ell$ that is necessary to
prevent generating the same POB again.

\paragraph{Quantified Inductive Trace.} A quantified monotone inductive
trace $\cT$ is a sequence of sets $Q_i$.
Each $Q_i$ is a set of pairs,
where for each pair $(\ell, \sigma)$ in $Q_i$, $\ell$ is a formula over $X$, possibly with
free variables, such that 
all free variables $\fvars(\ell)$ are of sort $\Int$, and $\sigma$ is a
substitution such that 
$\fvars(\ell) \subseteq \dom(\sigma)$ and
$\range(\sigma) \subseteq X' \cup \SK$. Intuitively, a pair
$(\ell, \sigma)$ corresponds to a universally quantified lemma
$\forall\ell$ and its ground instance $\ell\sigma$.
If $\ell$ has no free variables, it represents a ground lemma (as in
the original IC3).
 We write
$\forall Q_i = \{\forall L \mid (L, \sigma) \in Q_i\}$ for the set of
all ground and quantified lemmas in $Q_i$, and $\qi(Q_i) = \{\ell\sigma \mid
(\ell, \sigma) \in Q_i\}$ for the set of all instances in $Q_i$.

\textsc{Quic3} maintains that the trace $\cT$ is inductive and
monotone. That is, it satisfies the following conditions, where $N$ is
the size of $\cT$:
\begin{align*}
  \Init &\limp \forall Q_0 &
  \forall 0 \leq i < N \cdot \forall Q_i \land \Tr &\limp
  \forall Q_{i+1} &
  \forall Q_{i+1} \subseteq&\forall Q_i
\end{align*}
The first two conditions ensure inductiveness and the last ensures
syntactic monotonicity. Both are similar to the corresponding
conditions in \textsc{IC3}.

\paragraph{The rules.} The rules \textbf{Safe}, \textbf{Cex}, \textbf{Unfold},
\textbf{Candidate} 
are essentially the same as
their \textsc{IC3} counterparts. The only exception is that, whenever
the lemmas of frame $i$ are required, the instances $\qi(Q_i)$ of the
quantified lemmas in $Q_i$ are used (instead of $\forall Q_i$).
This ensures that the corresponding satisfiability
checks are decidable and do not diverge.

\paragraph{\textbf{Predecessor} rule.} \textbf{Predecessor} extends a POB
$\langle m, \xi, i+1\rangle \in \Queue$ from frame $i+1$ with a
predecessor POB $\langle \psi, \sigma, i \rangle$ at frame $i$. The
precondition to the rule is satisfiability of
$\qi(Q_i) \land \Tr \land (m'_{\sk})$. Note that all free
variables in the current POB $m$ are skolemized via the substitution
$\sk$ (recall that all the free variables are of sort $\Int$) and all constants are primed.

\textbf{Predecessor} rule extends the corresponding rule of
\textsc{IC3} in two ways.
First, POBs are generated using partial MBP. The
$\pMBP (X'\cup \SK, \Tr \land m'_{\sk}, M)$ is used to construct a
ground monomial $\varphi$
over $X \cup X' \cup \SK$, describing a
predecessor of $m'_{\sk}$. Whenever $\varphi$ contains constants from
$X' \cup SK$, these are abstracted by fresh free variables to
construct a POB $\psi$ over $X$. Thus, the newly constructed POB is not ground
and its free variables are implicitly existentially quantified.
(Since $\pMBP$ is guaranteed to eliminate all constants of sort $\Array$,
the free variables are all of sort $\Int$).
Second, the \textbf{Predecessor} maintains with the POB $\psi$ the
substitution $\sigma$ that corresponds to the inverse of the
abstraction used to construct $\psi$ from $\varphi$, i.e.,
$\psi\sigma = \varphi$.  
It is used to introduce a ground instance that blocks $\psi$
as a predecessor of $\langle m, \xi, i+1\rangle$ when the POB is blocked (see
\textbf{NewLemma}).

The soundness of \textbf{Predecessor} (in the sense that it does not introduce spurious counterexamples) rests on the fact that every
state in the generated POB has a $\Tr$ successor in the original
POB. This is formalized as follows: 
\begin{lemma}\label{lem:pred}
  Let
    $\langle m, \xi, i+1\rangle \in \Queue$ and let
    $(\psi, \sigma, i)$ be the POB computed by
    \textbf{Predecessor}.
    Then, 
    $(\exists \psi) \limp \exists X' \cdot (Tr \land \exists m')$.
\end{lemma}
\begin{proof}
  From the definition of \textbf{Predecessor},
  $(\psi, \sigma) = \abs(U, \varphi)$, where
  $(\varphi, U) = \pMBP (X'\cup \SK, \Tr \land m'_{\sk}, M)$. The set
  $U \subseteq X' \cup \SK$ are the constants that were not eliminated
  by MBP.  Then, by properties of $\pMBP$,
  $\psi \sigma \limp \exists (X', \SK) \setminus U \cdot Tr \land
  m'_{\sk}$.  Note that $(\exists U \cdot \varphi) = \exists \psi$. By
  abstracting $U$ in $\varphi$ and existentially quantifying over the
  resulting variables in both sides of the implication, we get that
  $\exists \psi \limp \exists X', \SK \cdot Tr \land m'_{\sk}$.  Since
  $\SK$ does not appear in $Tr$, the existential quantification
  distributes over $\Tr$:
  $\exists X', \SK \cdot Tr \land m'_{\sk} \equiv \exists X' \cdot (Tr
  \land \exists m')$.  \qed
\end{proof}

By induction and \cref{lem:pred}, we get that if
$ \langle \psi, \sigma, i\rangle$ is a POB in $\Queue$, then every
state satisfying $\exists\psi$ can reach a state in $\Bad$.

\paragraph{\textbf{NewLemma} rule.} \textbf{NewLemma} creates a potentially quantified lemma $\ell$ and a
corresponding instance $\ell\sigma$ to block a quantified POB
$\langle m, \sigma, i+1\rangle$ at level $i+1$. Note that if $\ell$ is quantified, then while the
instance $\ell\sigma$ is guaranteed to be new at level $i+1$, the lemma
$\ell$ might already appear in $Q_{i+1}$. The lemma $\ell$ is first
computed as in \textsc{IC3}, but using a skolemized version of the
POB. Second, if any skolem constants remain in the lemma, then they are
re-abstracted into the original variables. The corresponding instance of $\ell$
is determined by the substitution $\sigma$ of the POB.
Note that the instance $\ell \sigma$ is well defined since $\abs$ abstracts skolem constants back into the variables (of sort $\Int$) that introduced them,
ensuring that $\fvars(\ell) \subseteq \dom(\sigma)$.
Note further that if $\ell$ has no free variables, then the substitution $\sigma$ is redundant and could be replaced by an empty substitution.
(In fact, it is always sufficient to project $\sigma$ to $\fvars(\ell)$.)

The soundness of \textbf{NewLemma} follows form the fact that every
lemma $(\ell, \sigma)$ that is added to the trace $\cT$ keeps the
trace inductive. Formally:
\begin{lemma}\label{lem:new-q-lemma}
  Let $(\ell, \sigma)$ be a quantified lemma added to $Q_{i+1}$ by
  \textbf{NewLemma}. Then,
  $\cF(\forall Q_{i}) \limp (\forall \ell')$.
\end{lemma}
\begin{proof}
  $\ell$ is $\abs(\SK, L)$, where
  $L' = \Itp (\cF(\qi(Q_i)), m'_{\sk})$. Therefore,
  $\cF(\qi(Q_i)) \land \neg L'$ is unsatisfiable.  Let
  $\Psi$ be
  $\cF(\forall Q_{i}) \land (\neg \forall \ell')$, and
  assume, to the contrary, that $\Psi$ is satisfiable.  Since no
  constants from $\SK$ appear in 
  $\cF(\forall Q_{i})$ and
  $\ell$ is $\abs(\SK, L)$, $\Psi$ is equi-satisfiable to
  $\cF(\forall Q_{i}) \land (\neg L')$. Let $M$ be the
  corresponding model. Then, in contradiction,
  $M \models \cF(\qi(Q_i)) \land (\neg L')$.  \qed
\end{proof}

Rules \textbf{Predecessor} 
and \textbf{NewLemma}
use $m'_{\sk}$ that is skolemized with our special skolem substitution
where $\sk(\var{i}) = \sk_i$. We note that while the skolem
constants in $m'_{\sk}$ are always a subset of $\SK$ and do not
overlap with $X \cup X'$, they may overlap the existing skolem constants
that appear in the rest of the formula (e.g., if the rest of the
formula contains $\qi(Q_{i-1})$, where the ground instances result
from previously blocked POBs and, therefore, also contain skolem
constants). In this sense, our skolemization appears non-standard. However,
all the claims in this section only rely on the fact that the range of
$\sk$ is $\SK$ and that $\SK$ is disjoint from $X \cup X'$, which
holds for $\sk$.  

\paragraph{\textbf{Push} rule.}
 \textbf{Push} is similar to its
\textsc{IC3} counterpart.
It propagates a
(potentially quantified) lemma to the next frame. The key difference is the use of quantified formulas
$\forall Q_i$ (and their instantiations $\qi(Q_i)$ in the
pre-condition of the rule. Thus, checking applicability of
\textbf{Push} requires deciding validity of a
quantified FOL formula, which is undecidable in general. In practice,,
we use a weaker, but decidable, variant of these rules. In particular,
we use a finite instantiation strategy to instantiate $\forall Q_i$ in
combination with all of the instantiations $\qi(Q_i)$ discovered by
\textsc{Quic3} before theses rules are applied. This ensures progress
(i.e., \textsc{Quic3} never gets stuck in an application of a rule) at
an expense of completeness (some lemmas are not pushed as far as
possible, which impedes divergence).

\begin{figure}[t]
  \centering
  \lstset{basicstyle=\ttfamily\scriptsize}
  \lstset{language=C,morekeywords={assert,assume,nd}}
\begin{lstlisting}
void init_arrray(int[] A, int sz) {
  1: for (int i = 0; i < sz; i++)  A[i] = 0;
  2: j = nd(); assume(0 <= j && j < sz);
  3: assert(A[j] == 0);}
\end{lstlisting}
  \vspace{-0.2in}
  \caption{An array manipulating program.}
  \label{fig:examples}
\end{figure}

We illustrate the rules on a simple array-manipulating program
\texttt{init\_array} shown in \cref{fig:examples}. In the program,
\texttt{assume} and \texttt{assert} stand for the usual assume and
assert statements, respectively, and \texttt{nd} returns a
non-deterministic value. We assume that the program is converted into
a safety problem as usual.  In this problem, a special variable $pc$
is used to indicate the program counter. The first POB found by
\textbf{Candidate} is $pc=3 \land \sel(A, j) \neq 0$. Its predecessor,
is $pc = 2 \land \sel(A, v_0)\neq 0 \land 0 \leq v_0 < sz$
and the corresponding substitution is
$(v_0 \mapsto j)$. Note that since $\pMBP$ could not eliminate $j$, it
was replaced by a free variable. Eventually, this POB is blocked, the
lemma that is added is
$\forall ((pc = 2 \land 0 \leq v_0 < sz) \limp \sel(A, v_0)=0)$.

\paragraph{Soundness.}
 We conclude this section by showing that applying \textsc{Quic3}
rules from \cref{alg:lquic} in any order is sound:  
%
\begin{lemma}
  If \textsc{Quic3} returns \emph{Cex}, then $P$ is not safe (and
  there exists a counterexample).  Otherwise, if \textsc{Quic3}
  returns \emph{Safe}, then $P$ is safe.
\end{lemma}
\begin{proof}
  The first case follows immediately from \cref{lem:pred}.  The second
  case follows from the properties of the inductive trace maintained
  by \textsc{Quic3} that ensure that whenever \emph{Safe} is returned
  (by \textbf{Safe} rule), a safe inductive invariant is obtained.
  \cref{lem:new-q-lemma} ensures that these properties are preserved
  whenever a new quantified lemma is added.
  Soundness of all other rules follows the same argument as
  the corresponding rules of \textsc{IC3}. \qed
\end{proof}

In fact, \textsc{Quic3} ensures a stronger soundness guarantee:
\begin{lemma}
  In every step of \textsc{Quic3}, for every $k < N$,
  the sequence 
  $\{\forall Q_{i}\}_{i=1}^{k}$ is an interpolation sequence of length $k$ for $P$.
\end{lemma}
Thus, if \textsc{Quic3} reaches $N>k$, then there are no
counterexample of length $k$.

\section{Progress and Counterexamples}
\label{sec:progress-cex}

Safety verification of transition systems described in the theory of
LIA and Arrays is undecidable in general. Thus, there is no
expectation that \textsc{Quic3} always terminates.  None-the-less, it is
desirable for such a procedure to have strong progress guarantees
-- the longer it runs, the more executions are explored. In this
section, we show how to orchestrate the rules defining \textsc{Quic3}
(shown in \cref{alg:lquic}) into an effective procedure that
guarantees progress in exploration and produces a shortest
counterexample, if it exists. 

%
\newcommand{\Cex}{\textsc{Cex}}
\newcommand{\Proof}{\textsc{Safe}}
\newcommand{\Unknown}{\textsc{Unknown}}
\newcommand{\Blocked}{\textsc{Blocked}}
\newcommand{\Null}{\textsc{Null}}

\newcommand{\Intersect}{\textsc{Intersect}}

\SetKwFunction{qadd}{Add}
\SetKwFunction{qsize}{Size}
\SetKwFunction{qempty}{Empty}
\SetKwFunction{qtop}{Top}
\SetKwFunction{qpop}{Remove}
\SetKwFunction{parent}{Parent}
\SetKwFunction{block}{Block}
\SetKwFunction{qblock}{QBlock}
\SetKwFunction{sat}{SAT}

\SetKwFunction{quipRecBlockCube}{Quic3\_MakeSafe}
\SetKwFunction{quipPush}{Quic3\_Push}

\begin{figure}[t]
  \begin{algorithm}[H]
\DontPrintSemicolon
\LinesNumbered
\SetKw{Let}{let}
\SetKwFor{Forever}{forever}{do}{end}
\SetKw{Break}{break}
\SetKw{Continue}{continue}

\BlankLine
$N\gets 0$; $Q_0 = \{(\Init,\emptyset)\}$ \;
\uIf{$\Init \wedge \Bad$} {
  \Return \Cex\;
}
\While { (true) } {
  $N\gets N+1$; \quad $Q_N \gets \emptyset$ \;
  \uIf {$\quipRecBlockCube(\Bad, \emptyset, N) = \Cex$} {
     \Return \Cex\;
   }
  \uIf {$\quipPush() = \Proof$} {
     \Return \Proof\;
  }
}
\label{alg:quip-main}
\end{algorithm}
\vspace{-0.2in}
\caption{Main Procedure ($\texttt{Quic3\_Main}$). Wlog, we assume that
  $\Bad$ is a monomial.}
  \label{fig:quip-main}
\end{figure}

\begin{figure}[t]
\begin{algorithm}[H]
\DontPrintSemicolon
\LinesNumbered
\SetKw{Let}{let}
\SetKwFor{Forever}{forever}{do}{end}
\SetKw{Break}{break}
\SetKw{Continue}{continue}
\KwIn{(Cube $m_0$, Substitution $\sigma_0$, Level $i_0$)}
\KwData{
Queue $\Queue$ of triples $\langle m, \sigma, i \rangle$, where
  $m$ is a cube, $\sigma$ is a substitution and $i$ is a level }

$\Queue =\emptyset$\;
\tcp{Apply \textbf{Candidate} rule}
$\qadd(\Queue, \langle m_0, \sigma_0, i_0\rangle)$\;
\While{$\neg\qempty(\Queue)$} {
  $\langle m, \xi, i\rangle \gets \qtop(\Queue)$\;
\uIf{$i=0$} { 
      \tcp{Apply \textbf{Cex} rule; Found a counterexample} 
       \Return \Cex\;
 }
$M \gets \sat(\qi(Q_{i-1}) \wedge \Tr \wedge (m'_{\sk})) $ \;
    \uIf{$M \neq \bot$} {
        \tcp{Apply \textbf{Predecessor} rule}
        $(\varphi, U) \gets \pMBP (X'\cup \SK, \Tr \land m'_{\sk}, M)$ \;
        $(\psi, \sigma) \gets \abs(U, \varphi)$ \;
        $\qadd(\Queue,\langle \psi, \sigma, i-1 \rangle)$ \;
    }
    \uElse {
        $\qpop(\Queue, \langle m,\xi, i\rangle)$ \;
        $L' \gets \Itp (\qi(Q_{i-1} \wedge \Tr), m'_{\sk})$\;
        \tcp{Abstract all skolem constants}
        $(\ell, \_ ) \gets \abs(\SK, L)$\;
        \tcp{Optional quantified generalization (see Sec. 5)}
        $(\ell,\xi) \gets \texttt{QGen}(\ell, \langle m, \xi, i \rangle)$\;  \label{line:qgen}
        \tcp{Apply \textbf{NewLemma} rule}
        forall $j\leq i$, $Q_j \gets Q_j \cup \{(\ell, \xi)\}$\; \label{line:add-qlemma}
    }

}
\Return \Blocked

\label{alg:sis-new}
\end{algorithm}
\vspace{-0.1in}
  \caption{$\texttt{Quic3\_MakeSafe}$ procedure of \textsc{Quic3}.}
  \label{fig:quip-rec-block-cube}
\end{figure}

\paragraph{Realization of \textsc{Quic3}.} \cref{fig:quip-main}
depicts procedure $\texttt{Quic3\_Main}$ -- an instance of
\textsc{Quic3} where each iteration, starting from $N=0$, consists of
a $\texttt{Quic3\_MakeSafe}$ phase followed by a
$\texttt{Quic3\_Push}$ phase. The $\texttt{Quic3\_MakeSafe}$ phase,
described in \cref{fig:quip-rec-block-cube}, starts by initializing
$\Queue$ to the POB $(\Bad, \emptyset, N)$ (this is a degenerate
application of \textbf{Candidate} that is sufficient when $\Bad$ is a
monomial). It then applies \textbf{Predecessor} 
and \textbf{NewLemma} iteratively until either a counterexample is
found or $\Queue$ is emptied. 
\textbf{NewLemma} is preceded by an optional generalization procedure
(\cref{line:qgen}) that may introduce additional quantified variables
and record the constants that they originated from by extending the
substitution $\xi$.
We defer discussion of this procedure to \cref{sec:quant-gener}; in the simplest case, it will return the same lemma with the same substitution $\xi$.
At the end of $\texttt{Quic3\_MakeSafe}$, the
trace 
$(Q_i)_i$ is an interpolation sequence of length $N$.  The
$\texttt{Quic3\_Push}$ applies 
\textbf{Push}
iteratively from frame $i=1$ to $i=N$. The corresponding
satisfiability queries are restricted to use the existing instances of
quantified lemmas and a finite set of instantiations pre-determined by
heuristically chosen triggers.  If, as a result of pushing, two
consecutive frames become equal (rule \textbf{Safe}),
$\texttt{Quic3\_Main}$ returns \emph{Safe}.

\paragraph{Progress.} Recall that we use a \emph{deterministic}
skolemization procedure.  Namely, for a POB $\langle m,\xi,i \rangle$,
in every satisfiability check of the form
$\qi(Q_{i-1}) \wedge \Tr \wedge (m'_{\sk})$, the same
skolem substitution (defined by $\sk(v_i) =\sk_i$)
is used in $m'_{\sk}$, even if the rest of the formula (i.e., 
$\qi(Q_{i-1})$) changes.
The benefit of using a deterministic skolemization procedure is that
it ensures that all applications of $\pMBP$ in \textbf{Predecessor} use exactly the same formula $Tr\land m'_{\sk}$ and exactly the same set of constants. As a result,
the number of predecessors (POBs) generated by applications of
\textbf{Predecessor} for each POB is bounded by the finite range of
$\pMBP$ in its third (model) argument: 

\begin{lemma}
If a deterministic skolemization is used, then for each POB $\langle m,\xi,i \rangle$, the number of POBs generated by applying
\textbf{Predecessor} on $\langle m,\xi,i \rangle$ is finite.
\end{lemma}
\begin{proof}
For simplicity, we ignore the application of quantified generalization; the proof extends to handle it as well.
After a quantified lemma $(\ell,\xi)$ is added to $Q_{i-1}$,
every model 
$M \models \qi(Q_{i-1}) \land \Tr \land m'_{\sk}$ that is discovered when applying \textbf{Predecessor} on $\langle m,\xi,i \rangle$ will be such that $M \models \ell \xi$.
Recall that the lemma was generated by a POB $\langle \varphi, \sigma, i-1 \rangle$ that was blocked since
$\qi(Q_{i-2}) \land \Tr \land \varphi'_{\sk}$ was unsatisfiable, and $(\ell,\_) = \abs(\SK,L)$ where
$L' = \Itp (\qi(Q_{i-2} \wedge \Tr), \varphi'_{\sk})$. Therefore $L \wedge \varphi_{\sk} \equiv \bot$. Since $\abs$ maps each skolem constant back to the variable that introduced it, we have that the skolems in $L$ are abstracted to the original variables from $\varphi$. Hence, $\ell \wedge \varphi \equiv \bot$, which implies that $\ell\xi \wedge \varphi \xi \equiv \bot$. 
Thus, if
$M \models \qi(Q_{i-1}) \land \Tr \land m'_{\sk}$ then $M \not \models \varphi \xi$.
Therefore, $\pMBP(X'\cup\SK, Tr\land m'_{\sk}, M) \neq (\varphi \xi,\_ )$.
Meaning, once the POB that generated the lemma
was blocked, it cannot be rediscovered as a predecessor of $\langle m,\xi,i \rangle$.
Since the first two arguments of $\pMBP$ are the same in all applications of \textbf{Predecessor} on $\langle m,\xi,i \rangle$ (due to the deterministic skolemization), the finite range of $\pMBP$ implies that only finitely many predecessors are generated for the POB $\langle m,\xi,i \rangle$.
\qed
\end{proof}

Thus, for any value of $N$, there is only a finite number of POBs that
are added to $\Queue$ and processed by the rules, resulting in a
finite number of rule applications.  Moreover, since
$\texttt{Quic3\_Push}$ restricts the use of quantified lemmas to
existing ground instances and a finite instantiation scheme, and since
the other rules also use only these instances, all satisfiability
queries posed to the solver are of quantifier-free formulas in the
combined theories of LIA and Arrays, and as a result guaranteed to
terminate. This means that each rule is terminating.  Therefore,
$\texttt{Quic3\_Main}$ always makes progress in the following sense:

\begin{lemma} \label{lem:progress}
For every $k \in \nat$, $\texttt{Quic3\_Main}$ either reaches $N = k$, returns \emph{Safe},, or finds a counterexample.
\end{lemma}

\paragraph{Shortest Counterexamples.} $\texttt{Quic3\_Main}$ increases
$N$ only after an interpolation sequence of length $N$ is obtained, in
which case it is guaranteed that no counterexample up to this length
exists.  Combined with \cref{lem:progress} that ensures progress, this
implies that $\texttt{Quic3\_Main}$ always find a shortest
counterexample, if one exists:

\begin{corollary}
  If there exists a counterexample, then $\texttt{Quic3\_Main}$ is
  guaranteed to terminate and return a shortest counterexample.
\end{corollary}


\section{Quantified Generalization}
\label{sec:quant-gener}

\textsc{Quic3} uses quantified POBs to generate quantified lemmas.
However, these lemmas are sometimes too specific, hindering
convergence. This is addressed by \emph{quantified generalization}
(\qgen), a key part of \quic.  The \textsc{Quic3} rules in
\cref{alg:lquic} are extended with the rule \textbf{QGen} shown in
\cref{alg:qgen_rule}, and \texttt{Quic3\_MakeSafe}
(\Cref{fig:quip-rec-block-cube}) is extended with a call to
\texttt{QGen}, which implements \textbf{QGen}, before a new lemma is
added to its corresponding frame.

\begin{algorithm}[t]
    \begin{description}
    \setlength\itemsep{0.05in}
  \item[QGen] For $0 \leq i < n$ and a lemma
    $(\ell,\xi)\in Q_{i+1}$, let $g$ be a formula and $\sigma$ a
    substitution such that (i) $g \sigma \equiv \ell \xi$, (ii) $\fvars(\ell) \subseteq \fvars(g)$, and (iii)
    $\cF(\qi(Q_i)) \to \forall g'$. Then, add $(g, \sigma)$ to $Q_j$
    for all $0 \leq j \leq i+1$.
    \end{description}
    \caption{\textbf{QGen} rule for Quantified Generalization in \textsc{Quic3}.}
    \label{alg:qgen_rule}
    \vspace{-8pt}
\end{algorithm}

\paragraph{\textbf{QGen} rule.}
\textbf{QGen} generalizes a (potentially quantified) lemma
$(\ell,\xi) \in Q_{i+1}$ into a new quantified lemma $(g, \sigma)$
such that $(\forall g) \to (\forall\ell)$ is valid, i.e., the new
lemma $g$ is stronger than $\ell$.  The new quantified lemma $g$ and a
substitution $\rho$ (s.t. $g\rho \equiv \ell$) are constructed by
abstracting some terms of $\ell$ with fresh universally quantified
variables. If the new formula $\forall g$ is a valid lemma, i.e.,
$\cF(\qi(Q_i))\to\forall g'$ is valid, then \textbf{QGen} adds
$(g,\sigma)$ to $Q_j$ for $0\leq j\leq i+1$, where
$\sigma = \xi | \rho$. Note that the check ensures that the new lemma
maintains the interpolation sequence property of the trace. In the
rest of this section, we describe two heuristics to implement
\textbf{QGen} that we found useful in our benchmarks.

\paragraph{Simple \qgen} abstracts a single term in the input lemma
$\ell$ by introducing one \emph{additional} universally quantified
variable to $\ell$. In the new lemma $g$, the new variable $v$ appears
only as an index of an array (e.g., $\sel(A,v)$) or as an offset
(e.g., $\sel(A, i+v)$). Simple \qgen considers all $\sel$ terms in
$\ell$ and identifies sub-terms $t$ of index terms for which $\ell$
imposes lower and upper bounds. Each term $t$ is abstracted in turn
with bounds used as guards. For example, if $\ell$ is
$0 < sz \to (\sel(A, 0) = 42)$ and $t = 0$ of $\sel(A, 0)$, then a
candidate $(g, \sigma)$ is
$0 \leq v_0 < sz \to \sel(A, v_0) = 42$, and
$\{v_0 \mapsto 0\}$, where $v_0$ is universally quantified.

\paragraph{Arithmetic \qgen.}
Simple \qgen does not infer correlations neither between abstracted
terms nor between index and value terms. For example, it is unable to
create a lemma of the form
$\forall v\cdot 0\leq v < sz \to (\sel(A,v)= exp(v))$, where $exp(v)$
is some linear expression involving $v$. \emph{Arithmetic} \qgen
addresses this limitation by extracting and generalizing a correlation
between interpreted constants in the input lemma $\ell$. Arithmetic
\qgen works on lemmas $\ell$ of the form
$(\psi\land \phi_0\land\cdots\land\phi_{n-1}) \to \phi_n$, where there
is a formula $p(\vec{v})$ with free variables $\vec{v}$ and a set of
substitutions $\{\sigma_{k}\}_{k=0}^{n}$ s. t. $\phi_k = p
\sigma_k$. For example, $\ell$ is
$((1 < sz)\land (\sel(A,0) = 42))\to (\sel(A,1) = 44)$, where $p(i,j)$ is
$\sel(A, i) = j$, $\sigma_0$ is $\{i\mapsto 0, j\mapsto 42\}$, and
$\sigma_1$ is $\{i\mapsto 1, j\mapsto 44\}$. The substitutions can be
viewed as \emph{data points} and generalized by a convex hull, denoted
$ch$. For example, $ch(\{\sigma_0, \sigma_1\})$ =
$0\leq i \leq 1 \land j = 2i + 42$.
The lemma $\ell$ is strengthened by replacing the substitution of
$\phi_n$ with the convex hull by rewriting $\ell$ into
\mbox{$\forall \vec{v} \cdot (ch(\{\sigma_1, \ldots, \sigma_n\}) \land
  \psi \land \phi_0 \cdots \land \phi_{n-1}) \to p(\vec{v})$.}  In our
running example, this generates
$\forall i, j \cdot ( 0\leq i \leq 1 \land j = 2i + 42 \land 1 <
sz)\land (\sel(A,0) = 42))\to (\sel(A,i) = j)$. Note that only
$\phi_n$ is generalized, while all other $\phi_k$, $0 \leq k < n$,
provide the data points. Applying standard generalization might
simplify the lemma further by dropping $(\sel(A,0) = 42)$ and
combining $i \leq 1 \land 1 < sz$ into $1 < sz$, resulting in
$\forall i \cdot ( 0\leq i \leq sz)\to (\sel(A,i) = 2i+42)$. Note that
arithmetic \qgen applies to arbitrary linear arithmetic terms by
replacing the convex hull ($ch$) with the polyhedral join ($\sqcup$).

These two generalizations are sufficient for our benchmarks. However,
the power of \quic comes from the ability to integrate additional
generalizations, as required. For example, arithmetic \qgen can be
extended to consider not only a single lemma, but also mine other
existing lemmas for potential data points.

\section{Experimental Results}
\label{sec:experiments}

We have implemented \quic within the CHC engine of
Z3~\cite{DBLP:conf/tacas/MouraB08,DBLP:conf/cav/HoderBM11} and
evaluated it on 
array manipulating C programs from
SV-COMP~\cite{DBLP:conf/tacas/Beyer17} and
from~\cite{DBLP:conf/esop/DilligDA10}. We have converted C programs to
CHC using \textsc{SeaHorn}~\cite{DBLP:conf/cav/GurfinkelKKN15}. In
most of these examples, array bounds are fixed constants. We have
manually generalized array bounds to be symbolic to ensure that the
problems require quantified invariants. Note, however, that our
approach is independent of the value of the array bound (concrete or
symbolic). We stress that using \textsc{SeaHorn} prevents us from
using the ``best CHC encoding'' for a given problem, which is
unfortunately a common evaluation practice. By using \textsc{SeaHorn}
as is, we show how \textsc{Quic3} deals with complex realistic
intermediate representation. For example, \textsc{SeaHorn} generates
constraints supporting memory allocation and pointer arithmetic. This
complicates the necessary inductive invariants even for simple
examples. While we could have used a problem-specific encoding for
specially selected benchmarks, such an encoding does not uniformly
extend to all SV-COMP benchmarks.

Experiments were done on a Linux machine with an Intel E3-1240V2 CPU
and a timeout of 300 seconds. The source code for \textsc{Quic3} is
available in the main Z3 repository at
\url{https://github.com/Z3Prover/z3}. The CHC for all the benchmarks
are available at \url{https://github.com/chc-comp/quic3}. The results
for the safe instances -- the most interesting -- are shown in
Table~\ref{tbl:results}.  We compare with the \spacer engine of
Z3. \spacer supports arrays, but not quantifiers.  As expected,
\spacer times out on all of the benchmarks. We emphasize the
difference in the number of lemmas discovered by both
procedures. Clearly, since \quic discovers quantified lemmas, it
generates significantly fewer lemmas than \spacer. Each quantified
lemma discovered by \quic represents many ground lemmas that are
discovered by \spacer.

As shown in Table~\ref{tbl:results}, \quic times out on some of the
instances. This is due to a deficiency of the current implementation
of \qgen. Currently, \qgen only considers one candidate for abstraction,
and generalization fails if that candidate fails. Allowing \qgen
to try several candidates should solve this issue.

Unfortunately, we were unable to compare \quic to other related
approaches.  To our knowledge, tools that participated in SV-COMP~2018
are not able to discover the necessary quantified invariants and often
use unsound (i.e., bounded) inference. The closely related
tools, including \textsc{Safari}~\cite{DBLP:conf/cav/AlbertiBGRS12},
\textsc{Booster}~\cite{DBLP:conf/atva/AlbertiGS14},
and~\cite{DBLP:conf/esop/DilligDA10} are no longer available. Based on
our understanding of their heuristics, the invariants required in our
benchmarks are outside of the templates supported by these heuristics.

\begin{table}[t]
\centering
\caption{Summary of results. TO is timeout;
\emph{Depth} is the size of inductive trace;
\emph{Lemmas} and \emph{Inv} are the  number
of lemmas discovered overall and in \mbox{invariant, respectively.}}\label{tbl:results}
\scalebox{0.6}{
  \begin{tabular} {||l || c  | c | c | c || c | c ||}
    \hline
    \bf Benchmark &  \multicolumn{4}{|c||}{\textsc{\quic}} & \multicolumn{2}{|c||}{Z3/\textsc{Spacer}}\\
    \cline{2-7}
    & \bf Depth & \bf Lemmas & \bf Inv & \bf Time [s] & \bf Depth & \bf Lemmas\\
    \hline
    array-init-const & 6 & 24 &7 & 0.14 & 130 & 4,483 \\
array-init-partial & 9 & 45 & 12 & 0.34 & 126 & 4,224 \\
array-mono-set & 6 & 25 & 9 & 0.22 & 70 & 2,436 \\
array-mono-tuc & 6 & 25 & 9 & 0.21 & 70 & 2,422 \\
array-mul-init-tuc & 129 & 8,136 &  -- & TO & 131 & 8,393 \\
array-nd-2-c-true & 6 & 37 & -- & TO & 39 & 1,482 \\
array-reverse & 6 & 21 & 5 & 0.18 & 144 & 729 \\
array-shadowinit-tuc & 30 & 252 & -- & TO & 99 & 5,005 \\
array-swap & 13 & 136 & 64 & 6.38 & 45 & 2,700 \\
array-swap-twice & 14 & 155 & -- & TO & 45 & 2,991 \\
sanfoundry-02-tucg & 11 & 89 & 31 & 1.57 & 46 & 1,986 \\
sanfoundry-10-tucg & 11 & 71 & 23 & 0.67 & 109 & 3,245 \\
sanfoundry-27-tucg & 6 & 24 & 7 & 0.14 & 131 & 4,568 \\
std-compMod-tucg & 10 & 120 & 61 & 5.48 & 58 & 3,871 \\
std-copy1-tucg & 6 & 33 & 14 & 0.33 & 89 & 4,035 \\
std-copy2-tucg & 9 & 65 & 25 & 0.77 & 73 & 2,751 \\
std-copy3-tucg & 13 & 109 & 39 & 1.86 & 76 & 2,806 \\
std-copy4-tucg & 18 & 217 & -- & TO & 85 & 3,416 \\
std-copy5-tucg & 19 & 233 & 76 & 5.47 & 90 & 3,642 \\
std-copy6-tucg & 22 & 301 & -- & TO & 97 & 3,991 \\
std-copy7-tucg & 25 & 357 & -- & TO & 101 & 4,321 \\
std-copy8-tucg & 27 & 430 & 105 & 8.05 & 106 & 4,581 \\
  \hline
  \end{tabular}\hspace{0.4in}
\begin{tabular} {||l || c  | c | c | c || c | c ||}
    \hline
    \bf Benchmark &  \multicolumn{4}{|c||}{\textsc{\quic}} & \multicolumn{2}{|c||}{Z3/\textsc{Spacer}}\\
    \cline{2-7}
    & \bf Depth & \bf Lemmas & \bf Inv & \bf Time [s] & \bf Depth & \bf Lemmas\\
    \hline
std-copy9-tucg & 31 & 538 & 145 & 14.74 & 111 & 5,078 \\
std-copyInitSum2-tucg & 32 & 511 & -- & TO & 77 & 2,987 \\
std-copyInitSum3-tucg & 14 & 127 & -- & TO & 76 & 3,103 \\
std-copyInitSum-tucg & 9 & 59 & 21 & 0.43 & 78 & 3,085 \\
std-copyInit-tucg & 10 & 69 & 27 & 0.59 & 75 & 2,851 \\
std-find-tucg & 8 & 35 & 7 & 0.32 & 105 & 2,915 \\
std-init2-tucg & 7 & 29 & 8 & 0.14 & 88 & 3,662 \\
std-init3-tucg & 7 & 30 & 8 & 0.14 & 95 & 4,122 \\
std-init4-tucg & 7 & 31 & 8 & 0.14 & 94 & 3,898 \\
std-init5-tucg & 7 & 32 & 8 & 0.14 & 93 & 4,152 \\
std-init6-tucg & 7 & 33 & 8 & 0.15 & 95 & 4,090 \\
std-init7-tucg & 7 & 34 & 8 & 0.14 & 100 & 4,916 \\
std-init8-tucg & 7 & 35 & 8 & 0.15 & 97 & 4,604 \\
std-init9-tucg & 7 & 32 & 11 & 0.21 & 100 & 4,929 \\
std-maxInArray-tucg & 7 & 30 & 9 & 0.33 & 132 & 4,618 \\
std-minInArray-tucg & 7 & 30 & 10 & 0.27 & 133 & 4,686 \\
std-palindrome-tucg & 5 & 14 & -- & TO & 64 & 1,717 \\
std-part-orig-tucg & 10 & 83 & 11 & 11.59 & 138 & 5,035 \\
std-part-tucg & 13 & 103 & 41 & 1.7 & 132 & 4,746 \\
std-sort-N-nd-assert-L& 12 & 100 & 15 & 5.02 & 5 & 17 \\
std-vararg-tucg-tt & 9 & 40 & 10 & 0.23 & 133 & 4,622 \\
std-vector-diff-tucg & 12 & 112 & 14 & 2.94 & 76 & 2,964 \\

    \hline
\end{tabular}
}
\end{table}

\section{Related Work}
\label{sec:related}

Universally quantified invariants are necessary for verification of
systems with unbounded state size (i.e., the size of an individual
system state is unbounded) such as array manipulating programs,
programs with dynamic memory allocation, and parameterized systems in
general. Thus, the problem of universal invariant inference has been a
subject of intense research in a variety of areas of automated
verification. In this section, we present the related work that is
technically closest to ours and is applicable to the area of software
verification.

Classical predicate
abstraction~\cite{DBLP:conf/cav/GrafS97,DBLP:conf/tacas/BallPR01} has
been adapted to quantified invariants by extending predicates with
\emph{skolem} (fresh)
variables~\cite{DBLP:conf/popl/FlanaganQ02,DBLP:conf/vmcai/LahiriB04}. This
is sufficient for discovering complex loop invariants of array
manipulating programs similar to the ones used in our
experiments. These techniques require a decision procedure for
satisfiability of universally quantified formulas, and, significantly
complicate predicate
discovery~(e.g.,~\cite{DBLP:conf/cav/LahiriB04}). \textsc{Quic3}
extends this work to the \textsc{IC3} framework in which the predicate
discovery is automated and quantifier instantiation and instance
discovery are carefully managed throughout the procedure.

Recent
work~\cite{DBLP:conf/sas/BjornerMR13,DBLP:conf/sas/MonniauxG16,fse16}
studies this problem via the perspective of discovering universally
quantified models for CHCs. These works show that fixing the number of
expected quantifiers in an invariant is sufficient to approximate
quantified invariants by discovering a quantifier free invariant of a
more complex system. The complexity comes in a form of transforming
linear CHC to non-linear CHC (\emph{linear} refers to the shape of
CHC, not the theory of constraints). Unlike predicate abstraction,
guessing the predicates apriori is not required. However, both the
quantifiers and their instantiations are guessed eagerly based on the
syntax of the input problem. In contrast, \textsc{Quic3} works
directly on linear CHC (i.e., a transition system), and discovers
quantifiers and instantiations on demand. Hence, \textsc{Quic3} is not
limited to a fixed number of quantifiers, and, unlike these
techniques, is guaranteed to find the shortest counterexample. 

Model-Checking Modulo Theories
(MCMT)~\cite{DBLP:conf/cade/GhilardiR10} extends model checking to
array manipulating programs and has been used for verifying
heap manipulating programs and parameterized systems
(e.g.,~\cite{DBLP:conf/fmcad/ConchonGKMZ13}). It uses a combination
of quantifier elimination (QELIM) for computing predecessors of $\Bad$,
satisfiability checking of universally quantified formulas for pruning
exploration (and convergence check), and custom generalization
heuristics. In comparison, \textsc{Quic3} uses MBP
instead of QELIM and uses generalizations
based on bounded exploration.

\textsc{Safari}~\cite{DBLP:conf/cav/AlbertiBGRS12} (and later
\textsc{Booster}~\cite{DBLP:conf/atva/AlbertiGS14}), that extend MCMT
with Lazy Abstraction With Interpolation
(LAWI)~\cite{DBLP:conf/cav/McMillan06}, is closest to
\textsc{Quic3}. As in LAWI, interpolation (in case of \textsc{Safari},
for the theory of arrays~\cite{DBLP:journals/corr/abs-1204-2386}) is
used to construct a quantifier-free proof $\pi$ of bounded safety. The
proof $\pi$ is generalized by universally quantifying out some terms,
and a decision procedure for universally quantified formulas is used
to determine convergence. The key differences between \textsc{Safari}
and \textsc{Quic3} are the same as between \textsc{Lawi} and
\textsc{IC3}. We refer the reader to~\cite{DBLP:conf/cav/VizelG14} for
an in-depth comparison. Specifically, \texttt{Quic3\_MakeSafe}
computes an interpolation sequence that can be used for
\textsc{Safari}. However, unlike \textsc{Safari}, \textsc{Quic3} does
not rely on an external array interpolation procedure. Moreover, in
\textsc{Quic3}, the generalizations are dynamic and the
quantifiers are introduced as early as possible, potentially
exponentially simplifying the bounded proof. Finally,
\textsc{Quic3} manages its quantifier instantiations to
avoid relying on an external (semi) decision procedure. The
acceleration techniques used in \textsc{Booster} are orthogonal to
\textsc{Quic3} and can be combined in a form of pre-processing.

To our knowledge, \textsc{UPDR}~\cite{DBLP:conf/cav/KarbyshevBIRS15}
is the only other extension of IC3 to quantified invariants. The key
difference is that UPDR focuses on programs specified 
using the
Effectively PRopositional (EPR) fragment of \emph{uninterpreted} first order logic (e.g., without arithmetic) for
which quantified satisfiability is decidable. As such, UPDR does not
deal with quantifier instantiation and its mechanism for discovering quantifiers is different. 
\textsc{UPDR} is also limited to abstract
counterexamples (i.e., counterexamples to existence of universal
inductive invariants, as opposed to counterexamples to safety).

Interestingly, \textsc{Quic3} is closely related to algorithms for
quantified satisfiability
(e.g.,~\cite{DBLP:conf/vmcai/BradleyMS06,DBLP:conf/cav/GeM09,DBLP:conf/lpar/BjornerJ15}). \textsc{Quic3}
uses a MBP to construct a complete instantiation,
if possible. However,
unlike~\cite{DBLP:conf/vmcai/BradleyMS06,DBLP:conf/cav/GeM09}, the
convergence (of \texttt{Quic3\_MakeSafe}) does not rely on any
syntactic feature of the quantified formula.

\section{Conclusion}
\label{sec:conclusion}

In this paper, we present \quic, an extension of IC3 to reasoning
about array manipulating programs by discovering quantified inductive
invariants. While our extension keeps the basic structure of the IC3
framework, it significantly affects how lemmas and proof obligations
are managed and generalized. In particular, guaranteeing progress in
the presence of quantifiers requires careful management of the
necessary instantiations. Furthermore, discovering quantified lemmas,
requires new lemma generalization techniques that are able to infer
universally quantified facts based on several examples. Unlike
previous works, our generalizations and instantiations are done
\emph{on demand} guided by the property and current proof
obligations. We have implemented \quic in the CHC engine of Z3 and
show that it is competitive for reasoning about C programs.

\subsubsection*{Acknowledgements.}
This publication is part of a project that has received funding from
the European Research Council (ERC) under the European Union's Horizon
2020 research and innovation programme (grant agreement No
[759102-SVIS]).
The research was partially supported by Len Blavatnik and the Blavatnik Family
foundation, the Blavatnik Interdisciplinary Cyber Research Center, Tel Aviv
University, and the United States-Israel
Binational Science Foundation (BSF) grants No. 2016260 and 2012259.
We acknowledge the support of the Natural Sciences and Engineering
Research Council of Canada (NSERC), RGPAS-2017-507912.

\bibliographystyle{abbrv}
\bibliography{refs-short}

\begin{thebibliography}{10}

\bibitem{DBLP:conf/cav/AlbertiBGRS12}
F.~Alberti, R.~Bruttomesso, S.~Ghilardi, S.~Ranise, and N.~Sharygina.
\newblock {SAFARI:} {SMT}-based abstraction for arrays with interpolants.
\newblock In {\em {CAV}}, 2012.

\bibitem{DBLP:conf/atva/AlbertiGS14}
F.~Alberti, S.~Ghilardi, and N.~Sharygina.
\newblock Booster: An acceleration-based verification framework for array
  programs.
\newblock In {\em {ATVA}}, 2014.

\bibitem{DBLP:conf/tacas/BallPR01}
T.~Ball, A.~Podelski, and S.~K. Rajamani.
\newblock Boolean and cartesian abstraction for model checking {C} programs.
\newblock In {\em {TACAS}}, 2001.

\bibitem{DBLP:conf/tacas/Beyer17}
D.~Beyer.
\newblock Software verification with validation of results - (report on
  {SV-COMP} 2017).
\newblock In {\em {TACAS}}, 2017.

\bibitem{DBLP:conf/vmcai/BjornerG15}
N.~Bj{\o}rner and A.~Gurfinkel.
\newblock Property directed polyhedral abstraction.
\newblock In {\em {VMCAI'15}}, 2015.

\bibitem{DBLP:conf/lpar/BjornerJ15}
N.~Bj{\o}rner and M.~Janota.
\newblock Playing with quantified satisfaction.
\newblock In {\em {LPAR}}, 2015.

\bibitem{DBLP:conf/sas/BjornerMR13}
N.~Bj{\o}rner, K.~L. McMillan, and A.~Rybalchenko.
\newblock On solving universally quantified {H}orn clauses.
\newblock In {\em Static Analysis {(SAS)}}, 2013.

\bibitem{DBLP:conf/vmcai/Bradley11}
A.~R. Bradley.
\newblock {SAT-Based Model Checking without Unrolling}.
\newblock In {\em {VMCAI}}, 2011.

\bibitem{DBLP:conf/vmcai/BradleyMS06}
A.~R. Bradley, Z.~Manna, and H.~B. Sipma.
\newblock What's decidable about arrays?
\newblock In {\em Verification, Model Checking, and Abstract Interpretation
  {(VMCAI)}}, 2006.

\bibitem{DBLP:journals/corr/abs-1204-2386}
R.~Bruttomesso, S.~Ghilardi, and S.~Ranise.
\newblock Quantifier-free interpolation of a theory of arrays.
\newblock {\em Logical Methods in Computer Science}, 8(2), 2012.

\bibitem{DBLP:conf/fmcad/ConchonGKMZ13}
S.~Conchon, A.~Goel, S.~Krstic, A.~Mebsout, and F.~Za{\"{\i}}di.
\newblock Invariants for finite instances and beyond.
\newblock In {\em {FMCAD}}, 2013.

\bibitem{DBLP:conf/tacas/MouraB08}
L.~M. de~Moura and N.~Bj{\o}rner.
\newblock {Z3:} an efficient {SMT} solver.
\newblock In {\em {TACAS}}, 2008.

\bibitem{DBLP:conf/esop/DilligDA10}
I.~Dillig, T.~Dillig, and A.~Aiken.
\newblock Fluid updates: Beyond strong vs. weak updates.
\newblock In {\em European Symposium on Programming {(ESOP)}}, 2010.

\bibitem{DBLP:conf/popl/FlanaganQ02}
C.~Flanagan and S.~Qadeer.
\newblock Predicate abstraction for software verification.
\newblock In {\em {POPL}}, 2002.

\bibitem{DBLP:conf/cav/GeM09}
Y.~Ge and L.~M. de~Moura.
\newblock Complete instantiation for quantified formulas in satisfiabiliby
  modulo theories.
\newblock In {\em Computer Aided Verification {(CAV)}}, 2009.

\bibitem{DBLP:conf/cade/GhilardiR10}
S.~Ghilardi and S.~Ranise.
\newblock {MCMT:} {A} model checker modulo theories.
\newblock In {\em {IJCAR'10}}, 2010.

\bibitem{DBLP:conf/cav/GrafS97}
S.~Graf and H.~Sa{\"{\i}}di.
\newblock Construction of abstract state graphs with {PVS}.
\newblock In {\em {CAV'97}}, 1997.

\bibitem{DBLP:conf/fmcad/GurfinkelI15}
A.~Gurfinkel and A.~Ivrii.
\newblock Pushing to the top.
\newblock In {\em {FMCAD}}, 2015.

\bibitem{DBLP:conf/cav/GurfinkelKKN15}
A.~Gurfinkel, T.~Kahsai, A.~Komuravelli, and J.~A. Navas.
\newblock The {S}ea{H}orn verification framework.
\newblock In {\em Computer Aided Verification {(CAV)}}, 2015.

\bibitem{fse16}
A.~Gurfinkel, S.~Shoham, and Y.~Meshman.
\newblock {SMT}-based verification of parameterized systems.
\newblock In {\em {FSE}}, 2016.

\bibitem{DBLP:conf/sat/HoderB12}
K.~Hoder and N.~Bj{\o}rner.
\newblock Generalized property directed reachability.
\newblock In {\em {SAT'12}}, 2012.

\bibitem{DBLP:conf/cav/HoderBM11}
K.~Hoder, N.~Bj{\o}rner, and L.~M. de~Moura.
\newblock \emph{{\(\mu\)}Z}- an efficient engine for fixed points with
  constraints.
\newblock In {\em Computer Aided Verification {(CAV)}}, 2011.

\bibitem{DBLP:conf/cav/KarbyshevBIRS15}
A.~Karbyshev, N.~Bj{\o}rner, S.~Itzhaky, N.~Rinetzky, and S.~Shoham.
\newblock Property-directed inference of universal invariants or proving their
  absence.
\newblock In {\em {CAV}}, 2015.

\bibitem{DBLP:conf/fmcad/KomuravelliBGM15}
A.~Komuravelli, N.~Bj{\o}rner, A.~Gurfinkel, and K.~L. McMillan.
\newblock Compositional verification of procedural programs using {H}orn
  clauses over integers and arrays.
\newblock In {\em {FMCAD}}, 2015.

\bibitem{DBLP:conf/cav/KomuravelliGC14}
A.~Komuravelli, A.~Gurfinkel, and S.~Chaki.
\newblock {SMT-Based Model Checking for Recursive Programs}.
\newblock In {\em Computer Aided Verification {(CAV)}}, 2014.

\bibitem{DBLP:conf/vmcai/LahiriB04}
S.~K. Lahiri and R.~E. Bryant.
\newblock Constructing quantified invariants via predicate abstraction.
\newblock In {\em {VMCAI}}, 2004.

\bibitem{DBLP:conf/cav/LahiriB04}
S.~K. Lahiri and R.~E. Bryant.
\newblock Indexed predicate discovery for unbounded system verification.
\newblock In {\em Computer Aided Verification {(CAV)}}, 2004.

\bibitem{DBLP:conf/cav/McMillan06}
K.~L. McMillan.
\newblock Lazy abstraction with interpolants.
\newblock In {\em {CAV}}, 2006.

\bibitem{DBLP:conf/sas/MonniauxG16}
D.~Monniaux and L.~Gonnord.
\newblock Cell morphing: From array programs to array-free {H}orn clauses.
\newblock In {\em Static Analysis {(SAS)}}, 2016.

\bibitem{DBLP:conf/cav/VizelG14}
Y.~Vizel and A.~Gurfinkel.
\newblock Interpolating property directed reachability.
\newblock In {\em {CAV'14}}, 2014.

\end{thebibliography}

\end{document}